\documentclass[10pt,journal,compsoc]{IEEEtran}

\ifCLASSOPTIONcompsoc
      \usepackage[nocompress]{cite}
\else
    \usepackage{cite}
\fi
\usepackage{graphicx}
\usepackage{balance}

\usepackage{Definitions}
\usepackage[utf8]{inputenc}
\usepackage[font=small,labelfont=bf,tableposition=top]{caption}
\usepackage{epstopdf}
\usepackage{subfig} 
\usepackage[export]{adjustbox}
\usepackage{booktabs} 
\usepackage{bbm}
\DeclareCaptionLabelFormat{andtable}{#1~#2  \&  \tablename~\thetable}

\newcommand{\xhdr}[1]{\vspace{0.2mm}\noindent{{\bf #1.}}}

\usepackage{xspace}
\usepackage{xfrac}

\usepackage{array}
\usepackage{tabularx, makecell, booktabs}

\newcommand{\tcim}{\textsc{TCIM}\xspace}
\newcommand{\tcimbudget}{\textsc{TCIM-Budget}\xspace}
\newcommand{\tcimcover}{\textsc{TCIM-Cover}\xspace}
\newcommand{\fairtcimbudget}{\textsc{\allowbreak FairTCIM-Budget}\xspace}
\newcommand{\fairtcimcover}{\textsc{\allowbreak FairTCIM-Cover}\xspace}

\newcommand{\ic}{\textsc{IC}\xspace}
\newcommand{\lt}{\textsc{LT}\xspace}

\begin{document}

\title{On the Fairness of Time-Critical Influence Maximization in Social Networks}

  \author{$\qquad$Junaid Ali$^\dagger$, Mahmoudreza Babaei$^\ddagger$, Abhijnan Chakraborty$^*$, Baharan Mirzasoleiman$^{**}$, \\ Krishna P. Gummadi$^\dagger$, Adish Singla$^\dagger$ \\ 
 $^\dagger$Max Planck Institute for Software Systems,  \\
 $^\ddagger$Max Planck Institute for Human Developement, \\  $^*$Indian Institute of Technology Delhi, $^{**}$UCLA
  }

\maketitle

\begin{abstract} 
Influence maximization has found applications in a wide range of real-world problems, for instance, viral marketing of products in an online social network, and propagation of valuable information such as job vacancy advertisements. While existing algorithmic techniques usually aim at maximizing the total number of people influenced, the population often comprises several socially salient groups, e.g., based on gender or race. As a result, these techniques could lead to disparity across different groups in receiving important information. Furthermore, in many applications, the spread of influence is time-critical, i.e., it is only beneficial to be influenced before a deadline. As we show in this paper, such time-criticality of information could further exacerbate the disparity of influence across groups. This disparity could have far-reaching consequences, impacting people's prosperity and putting minority groups at a big disadvantage. In this work, we propose a notion of \emph{group fairness} in \emph{time-critical influence maximization}. We introduce surrogate objective functions to solve the influence maximization problem under fairness considerations. By exploiting the submodularity structure of our objectives, we provide computationally efficient algorithms with guarantees that are effective in enforcing fairness during the propagation process. Extensive experiments on synthetic and real-world datasets demonstrate the efficacy of our proposal.
\end{abstract}


\vspace{-6mm}
\section{Introduction}
\vspace{-1mm}
The problem of {\it Influence Maximization} has been widely studied due to its application in multiple domains such as viral marketing~\cite{richardson2002mining}, social recommendations~\cite{ye2012exploring}, propagation of information related to jobs, financial opportunities or public health programs~\cite{banerjee2013diffusion,yadav2016using}. Over the years, extensive research efforts have focused on the cascading behavior, diffusion and spreading of ideas, or containment of diseases~\cite{leskovec2007cost, kempe2003maximizing, richardson2002mining, wallinga2004different, leskovec2010inferring}. The idea is to identify a set of initial sources (i.e., {\it seed nodes}) in a social network who can influence other people (e.g., by propagating key information), and traditionally the goal has been to maximize the total number of people influenced in the process (e.g., who received the information being propagated)~\cite{ kempe2003maximizing, goyal2013minimizing,  carnes2007maximizing}. 

Real-world social networks, however, are often not homogeneous and comprise different groups of people. Due to the disparity in their population sizes, potentially high propensity towards creating within-group links \cite{mcpherson2001birds}, and differences in dynamics of influences among different groups \cite{DBLP:conf/wsdm/SinglaW09}, the structure of the social network can cause disparities in the influence maximization process. For example, selecting most of the seed nodes from the majority group might maximize the total number of influenced nodes, but very few members of the minority group may get influenced. In many application scenarios such as propagation of job or health-related information, such disparity can end up impacting people's livelihood and some groups may become impoverished in the process.

Moreover, some applications are also {\it time-critical} in nature~\cite{chen2012time}. For example, many job applications typically have a deadline by which one needs to apply; if information related to the application reaches someone after the deadline, it is not useful. Similarly, in viral marketing, many companies offer discount deals only for few days (hours); getting this information late does not serve the recipient(s). More worryingly, if one group of people gets influenced (i.e., they get the information) faster than other groups, it could end up exacerbating the inequality in information access. This is possible if the majority group is better connected and more central in the network than the minority group. Thus, in time-critical application scenarios, focusing on the traditional criteria of maximizing the number of influenced nodes can have a disparate impact on different groups. This disparity in time-critical applications, in turn, can put minority and under-represented groups at a big disadvantage with far-reaching consequences.
In this paper, we attempt to mitigate such unfairness in time-critical influence maximization (\tcim), and we focus on two settings: (i) where the budget (i.e., the number of seeds) is fixed and the goal is to find a seed set which maximizes the time-critical influence, we call this as \tcimbudget problem, and (ii) where a certain quota or fraction of the population should be influenced under the prescribed time deadline, and the goal is to find such a seed set of minimal size, we call this as \tcimcover problem.

\vspace{-4mm}
\subsection{Our Contributions} 
Our first contribution is to formally introduce the notion of fairness in time-critical influence maximization, which requires that {\it within a prescribed time deadline, the fraction of influenced nodes should be equal across different groups}. We highlight, via experiments and an illustrative example, that the standard algorithmic techniques for solving \tcimbudget and \tcimcover problems lead to unfair solutions, and the disparity across groups could get worse with tighter time deadline. 
Secondly, we study the effect of disparity of influence  between groups: (i) by varying graph properties, such as connectivity and relative group sizes etc., and (ii) by varying \tcim algorithmic properties, such as seed budget, reach quota and time deadline etc.

We introduce two formulations of \tcim problems under fairness considerations, namely \fairtcimbudget and \fairtcimcover. 
As our third contribution, we propose {\it monotone submodular} surrogates for solving both of these NP-Hard problems. Though the surrogate problems are still NP-Hard, we propose a greedy approximation with provable guarantees.

We evaluate our proposed solutions over several synthetic and two real-world social networks and show that they are successful in enforcing the aforementioned fairness notion. Enforcing fairness does come at the cost of a reduction in performance.
However, as guaranteed by our theoretical results, our experiments indeed demonstrate that this cost of fairness, i.e.,  reduction in performance, is bounded for our approach.

\vspace{-5.0mm}
\section{Related Work}
In this section, we briefly review the related literature on influence maximization and algorithmic fairness.

\xhdr{Influence Maximization}
Richardson et al.~\cite{richardson2002mining} first introduced Influence Maximization as an algorithmic problem, and proposed a heuristic approach to find a set of nodes whose initial adoption of a certain idea/product can maximize the number of further adopters. Over the years, extensive research efforts have focused on the cascading behavior, diffusion and spreading of ideas or containment of diseases, by identifying the set of influential nodes that maximizes the influence through a network (often in real-time) 
~\cite{leskovec2007cost, kempe2003maximizing, richardson2002mining, wallinga2004different, leskovec2010inferring}. 

Typically, identifying the most influential nodes is studied in two ways: (i) using network structural properties to find the set of most central nodes \cite{kourtellis2013identifying,kempe2003maximizing}, 
and (ii) formulating the problem as discrete optimization \cite{goyal2013minimizing, kempe2003maximizing, babaei2013revenue}. 
Kempe et al. \cite{kempe2003maximizing}, studied influence maximization under different social contagion models and showed that submodularity of the influence function can be used to obtain provable approximation guarantees. Since then, there has been a large body of work studying various extensions \cite{goyal2013minimizing,bharathi2007competitive,budak2011limiting,carnes2007maximizing, huang2017revisiting}.  However, the notion of fairness in the influence maximization problem has not been studied by this line of previous works.

\vspace{1mm}
\xhdr{Fairness in Algorithmic Decision Making}
Recently a growing amount of work has focused on bias and unfairness in algorithmic decision-making systems \cite{ dwork2012fairness,feldman2015certifying, hardt2016equality}. The aim here is to examine and mitigate unfair decisions that may lead to discrimination. 
Although fairness along different dimensions of political science, moral philosophy, economics, and law has been extensively studied~\cite{conitzer2019group, fain2018fair, segal2017democratic, suksompong2018approximate} , only a few contemporary works have investigated fairness in influence maximization, as described next.

\vspace{1mm}
\xhdr{Contemporary Works}
Very recently, Fish et al. \cite{fish_dmt}, proposed a notion of individual fairness in information access, but did not consider the group fairness aspects. In addition, some prior works have proposed constrained optimization problems to encourage diversity in selecting the most influential nodes~\cite{bredereck2018multiwinner,faliszewski2017multiwinner,benabbou2018diversity,aghaei2019learning}. 

A recent paper by Rahmattalabi et al. \cite{rahmattalabi2019exploring}, proposes group fairness in influence maximization for robust covering problems. This method is different from ours in the following ways: i) their notion of fairness is maximizing the minimum influence for any group, while we propose parity of influence among different groups; ii) they consider a setting where seeds could be deactivated randomly while we do not have any stochasticity in seed activation; iii) they consider seed nodes to spread influence only to their immediate neighbors, while we vary the allowed time deadline and show its effect on disparity among different groups. We also demonstrate the effectiveness of our methods for different time deadlines on several datasets; iv) they propose an integer linear programming set up while we propose submodular proxies, akin to the traditional methods, which can be approximately solved using the greedy heuristic.

In concurrent works, Khajehnejad et al., \cite{khajehnejad2020adversarial}, and Tsang et al., \cite{Tsang_dmt}, proposed methods to achieve group fairness in influence maximization. 
However, their works are very different from our approach in three ways: i) they propose a different problem formulation with objective that does not have submodular structural properties, ii) they only study the problem under budget constraint, and iii) they do not consider the time-critical aspect of influence in their definition of fairness for influence maximization. This could result in majority groups being influenced before the minority, and can lead to disparity in applications where the timing of being influenced/informed is critical. In our work, we introduce a submodular objective that directly addresses the time-criticality in influence maximization problem under budget constraint as well as coverage constraint.

\begin{figure*}[t!]
        \begin{minipage}[t]{.25\linewidth}
    \includegraphics[width=\linewidth]{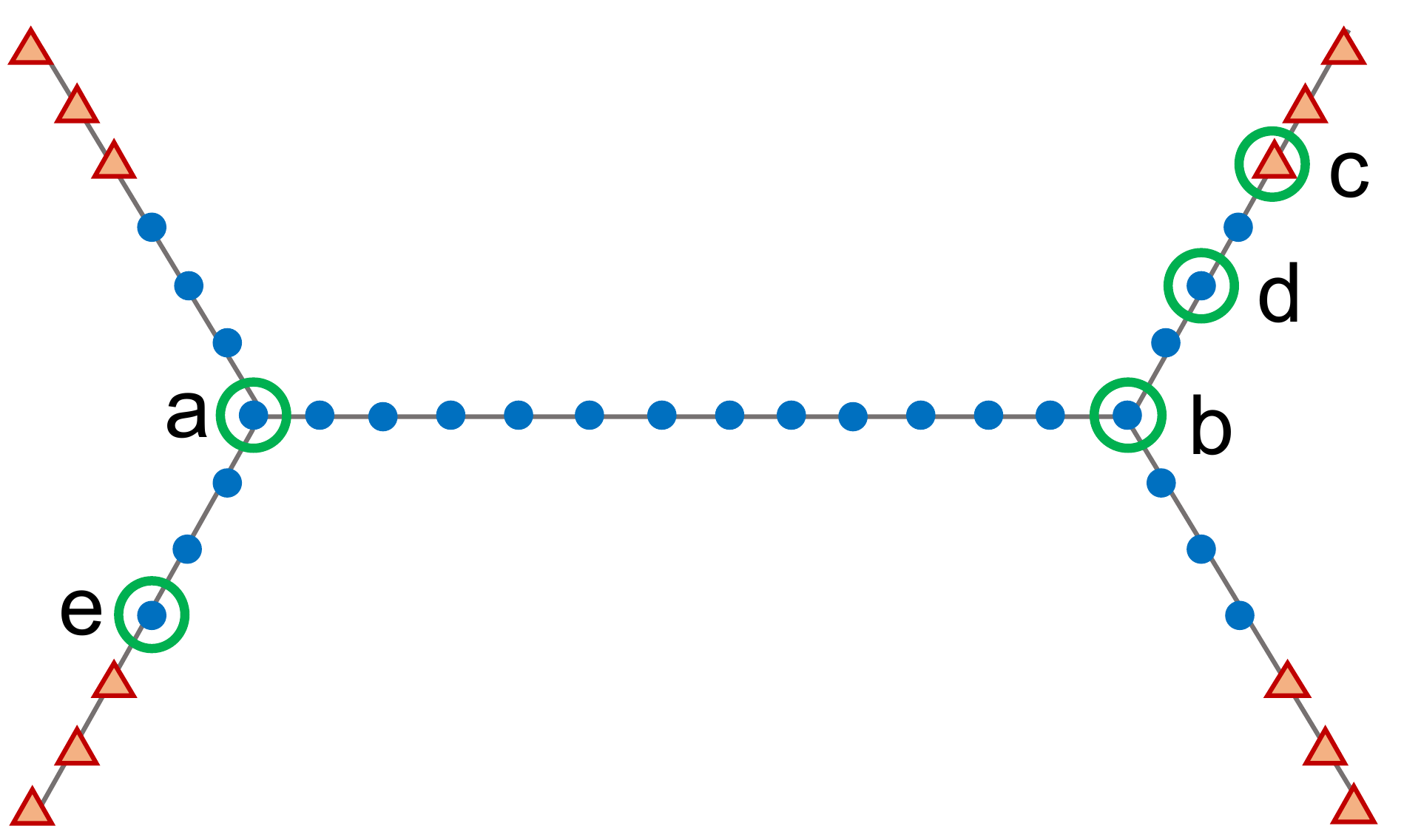}
    \label{fig:example_graph}
        \end{minipage}
        \begin{minipage}[b]{0.75\linewidth}
    \begingroup
    \resizebox{\linewidth}{!}
    {
	\renewcommand{\arraystretch}{1.2}
	\small
		\begin{tabular}[b]{c||cccc||cccc}\hline
      \toprule
      &  \multicolumn{4}{c||}{solution to \tcimbudget~\ref{eq:budget_problem}} &   \multicolumn{4}{c}{solution to \fairtcimbudget~\ref{eq:surrogate_fair_budget_problem}} \\
      \midrule
      & $S$ & $\frac{f(S; \Vcal, \Gcal)}{|\Vcal|}$ & $\frac{f(S; \Vcal_1, \Gcal)}{|\Vcal_1|}$ & $\frac{f(S; \Vcal_2, \Gcal)}{|\Vcal_2|}$ &$S$ & $\frac{f(S; \Vcal, \Gcal)}{|\Vcal|}$ & $\frac{f(S; \Vcal_1, \Gcal)}{|\Vcal_1|}$ & $\frac{f(S; \Vcal_2, \Gcal)}{|\Vcal_2|}$ \\
      \midrule
            $\tau=\infty$ & $\{a, b\}$ & 0.38 & 0.48 & 0.16 & $\{a, c\}$ &0.31 & 0.33 & 0.27 \\ 
      $\tau=4$       & $\{a, b\}$ &0.32 & 0.44 & 0.08 & $\{d, e\}$ &0.25 & 0.26 & 0.22 \\ 
      $\tau=2$       & $\{a, b\}$ &0.24 & 0.36 & 0.00 & $\{a, c\}$ &0.21 & 0.22 & 0.18 \\ 
      \bottomrule
    \end{tabular}
    
    }
    \endgroup
    \end{minipage}
    \captionlistentry[table]{Motivating Example}
          \vspace{-2mm}
    \caption{\looseness-1 An example to illustrate the disparity across groups in the standard approaches to \tcim. (Left) Graph with $|\Vcal| = 38$ nodes belonging to two groups shown in ``blue dots" ($|\Vcal_{1}| = 26$) and ``red triangles" ($|\Vcal_{2}| = 12$). (Right) We compare an optimal solution to the standard \tcimbudget problem~\ref{eq:budget_problem} and an optimal solution to our formulation of \tcimbudget with fairness considerations given by \fairtcimbudget problem~\ref{eq:surrogate_fair_budget_problem}. For different time critical deadlines $\tau$, normalized utilities are reported for the whole population $\Vcal$, for the ``blue dots" group $\Vcal_1$, and for the ``red triangles" group $\Vcal_2$. As $\tau$ reduces, the disparity between groups is further exacerbated in the solution to \tcimbudget problem~\ref{eq:budget_problem}. Solution to \fairtcimbudget problem~\ref{eq:surrogate_fair_budget_problem} achieves high utility and low disparity for different deadlines $\tau$.}
          \label{fig:example_graph}
  \vspace{-4mm}
  \end{figure*}
\vspace{-4mm}
\section{Background on Time-Critical Influence Maximization (TCIM)}\label{sec:background}
In this section, we provide the necessary background on the problem of time-critical influence maximization (henceforth, referred to as \tcim for brevity). First, we formally introduce a well-studied influence propagation model and specify the notion of time-critical influence that we consider in this paper. Then, we discuss two discrete optimization formulations to tackle the \tcim problem.

\vspace{-2mm}
\subsection{Influence Propagation in Social Network}
Consider a directed graph $\Gcal = (\Vcal, \Ecal)$, where $\Vcal$ is the set of nodes and $\Ecal$ is the set of directed edges connecting these nodes. For instance, in a social network the nodes could represent people and edges could represent friendship links between people.  An undirected link between two nodes can be represented by simply considering two directed edges between these nodes. 

There are two classical influence propagation models that are studied in the literature \cite{kempe2003maximizing}: (i) Independent Cascade model (\ic) and (ii) Linear Threshold (\lt) model. In this paper, we will consider \ic  model and our results can easily be extended to the \lt model. 

In the \ic model, there is a probability of influence associated with each edge denoted as $p_{\Ecal} := \{p_e \in [0, 1]: e \in \Ecal\}$. Given an initial seed set $S \subseteq \Vcal$, the influence propagation proceeds in discrete time steps $t= \{0, 1, 2, \ldots, \}$ as follows. At $t=0$, the initial seed set $S$ is ``activated" (i.e., influenced). Then, at any time step $t > 0$, a node $v \in \Vcal$ which was activated at time $t-1$ gets a chance to influence its neighbors (i.e., set of nodes $\{w : (v, w) \in \Ecal\}$). The influence propagation process stops at time $t > 0$ if no new nodes get influenced at this time. Under the \ic model, once a node is activated it stays active throughout the process and each node has only one chance to influence its neighbors. 

\looseness-1
Note that the influence propagation under \ic model is a stochastic process: the stochasticity here arises because of the random outcomes of a node $v$ influencing its neighbor $w$ based on the Bernoulli distribution $p_{(v, w)}$. An outcome of the influence propagation process can be denoted via a set of timestamps $\{t_v \geq 0 : v \in \Vcal\}$ where $t_v$ represents the time at which a node $v \in \Vcal$ was activated. We have  $t_v = 0 \textnormal{ iff } v \in S$ and for convenience of notation, we define $t_v = -1$ to indicate that the node $v$ was not activated in the process.

\vspace{-2.5mm}
\subsection{Utility of Time-Critical Influence}
As mentioned earlier, we focus on the application settings where the spread of influence is time-critical, i.e., it is more beneficial to be influenced earlier in the process. In particular, we adopt the well-studied notion of time-critical influence as proposed by \cite{chen2012time}. Their time-critical model is captured via a deadline $\tau$: If a node is activated before the deadline, it receives a utility of $1$, otherwise it receives no utility. This simple model captures the notion of timing in many important real-world applications such as viral marketing  of an online product with limited availability, information propagation of job vacancy information, etc.

Given the influence propagation model and the notion of time-critical aspect via a deadline $\tau$, we quantify the utility of time-critical influence for a given seed set $S$ on a set of target nodes $Y \subseteq \Vcal$ via the following: 
\begin{equation}\label{eq:time_critical_utility}
 f_{\tau} (S; Y, \Gcal) = \EE \Big[ \sum_{v \in Y, t_v \geq 0} \II(t_v \le \tau) \Big],	
\end{equation}
where the expectation is w.r.t. the randomness of the outcomes of the \ic  model.  The function is parametrized by deadline $\tau$, set $Y \subseteq \Vcal$ representing the set of nodes over which the utility is measured (by default, one can consider $Y = \Vcal$), and the underlying graph $\Gcal$ along with edge activation probabilities $p_{\Ecal}$. Given a fixed value of these parameters, the utility function $f_{\tau}: 2^\Vcal \rightarrow \mathbb{R}_{\geq 0}$ is a set function defined over the  seed set $S \subseteq \Vcal$.  Note that the constraint $t_v \geq 0$ represents the node was activated and the constraint $t_v \le \tau$ represents that the activation happened before the deadline $\tau$.

\vspace{-4mm}
\subsection{TCIM as Discrete Optimization Problem}
Next, we present two settings under which we study \tcim by casting it as a discrete optimization problem.

\subsubsection{Maximization under Budget Constraint (\tcimbudget)}
In the maximization problem under budget constraint, we are given a fixed budget $B > 0$ and the goal is to find an optimal set of seed nodes that maximize the expected utility. Formally, we state the problem as 
\begin{align}\label{eq:budget_problem}
  \max_{S \subseteq \Vcal}{f_{\tau}(S; \Vcal, \Gcal)} 
  \quad \text{ subject to } |S| \leq B. \tag{P1} 
\end{align}

\subsubsection{Minimization under Coverage Constraint (\tcimcover)}
In the minimization problem under coverage constraint, we are given a  quota $Q \in [0, 1]$ representing the minimal fraction of nodes that must be activated or ``covered" by the influence propagation in expectation. The goal is then to find an optimal set of seeds of minimal size that achieves the desired coverage constraint. We formally state the problem as 
\begin{align}\label{eq:cover_problem}
 \min_{S \subseteq \Vcal} |S| 
  \quad \text{ subject to }  \frac{f_{\tau}(S; \Vcal, \Gcal)}{|\Vcal|} \geq Q. \tag{P2} 
\end{align}
\vspace{-4.5mm}
\subsection{Submodularity and Approximate Solutions}\label{sec:background:submodular}
Next, we present some key properties  of the utility function $f_\tau(.)$ to get a better understanding of the above-mentioned optimization problems.  In their seminal work, \cite{kempe2003maximizing} showed that the utility function without time-critical deadline, i.e., $f_\infty(.): S \rightarrow \mathbb{R}_{+}$, is a non-negative, monotone,  submodular set function w.r.t. the optimization variable $S \subseteq \Vcal$. Submodularity is an intuitive notion of diminishing returns and optimization of submodular set functions finds numerous applications in machine learning and social networks, such
as influence maximization \cite{kempe2003maximizing}, sensing \cite{DBLP:conf/aaai/KrauseG07}, information gathering \cite{DBLP:conf/ijcai/SinglaHKWK15}, and active learning \cite{DBLP:conf/uai/GuilloryB11} (see \cite{krause2014submodular} for a survey on submodular function optimization and its applications).

Chen et. al \cite{chen2012time} showed that the utility function for the general time-critical setting for any $\tau$ also satisfies these properties.  Submodularity is an intuitive notion of diminishing returns, stating that, for any sets $A \subseteq A' \subseteq \Vcal$, and any node $a \in \Vcal \setminus A'$, it holds that (omitting the parameters $\Vcal$ and $\Gcal$ for brevity):
\begin{align*}
f_\tau(A \cup \{a\}) - f_\tau(A) \geq f_\tau(A' \cup \{a\}) - f_\tau(A').
\end{align*}

Existing works \cite{1978-_nemhauser_submodular-max,1998-_feige_threshold-of-ln-n,krause2014submodular} have shown that \ref{eq:budget_problem} and \ref{eq:cover_problem} are NP-Hard and hence finding the optimization solution is intractable. However, on a positive note, one can exploit the submodularity property of the function to design efficient approximation algorithms with provable guarantees~\cite{1978-_nemhauser_submodular-max,krause2014submodular}. In particular,  we can run the following greedy heuristic: start from an empty set, iteratively add a new node to the set that provides the maximal marginal gain in terms of utility, and stop the algorithm when the desired constraint on budget or coverage is met. This greedy algorithm provides the following guarantees for these two problems:.
\begin{itemize}
\item for the \tcimbudget problem~\ref{eq:budget_problem}, the greedy algorithm returns a set $\hat{S}$ that guarantees the following lower bound on the utility: $f_\tau(\hat{S}; \Vcal, \Gcal) \geq (1 - \frac{1}{e}) \cdot f_\tau(S^*; \Vcal, \Gcal)$ where $S^*$ is an optimal solution to problem~\ref{eq:budget_problem}.
\item for the \tcimcover problem~\ref{eq:cover_problem}, the greedy algorithm returns a set $\hat{S}$ that guarantees the following upper bound on the seed set size: $|\hat{S}| \leq \ln(1+ |\Vcal|) \cdot |S^*|$  where $S^*$ is an optimal solution to problem~\ref{eq:cover_problem}.
\end{itemize}

\vspace{-5mm}
\section{Measuring Unfairness in TCIM}\label{sec:notions}
In this section, we highlight the disparity in utility across population resulting from the solution to the standard \tcim problem formulations, and introduce a measure of unfairness in \tcim.

\vspace{-4mm}
\subsection{Socially Salient Groups and Their Utilities}\label{sec:notions:groups} 
The current approaches to \tcim consider all the nodes in $\Vcal$ to be homogeneous.
We capture the presence of different socially salient groups in the population by dividing individuals into $k$ disjoint groups. Here, socially salient groups could be based on some sensitive attribute such as gender or race. We denote the set of nodes in each group $i \in \{1, 2, \ldots, k\}$ as $\Vcal_i \subseteq \Vcal$, and we have $\Vcal = \cup_{i} \Vcal_i$. For any given seed set $S$, we define the utilities for a group $i$ as $f_{\tau} (S; \Vcal_{i}, \Gcal)$ by setting target nodes $Y = \Vcal_{i}$  in Eq.~\ref{eq:time_critical_utility}.

\vspace{-3mm}
\subsection{Disparity in Utility Across Groups}\label{sec:notions:example}
In the standard formulations for \tcim problem, i.e., \tcimbudget problem~\ref{eq:budget_problem} and \tcimcover problem~\ref{eq:cover_problem}, the utility $f_{\tau} (S; \Vcal, \Gcal)$ is optimized for the whole population $ \Vcal$ without considering their groups. Clearly, a solution to \tcim problem can, in general, lead to high disparity in utilities of different groups. 

In particular, this disparity in utility across groups arises from several factors in which two groups differ from each other. One of the factors is that the groups are of different sizes, i.e., one group is a minority. The different group sizes could, in turn, lead to selecting seed nodes from the majority group when optimizing for utility $f_{\tau} (S; \Vcal, \Gcal)$ in problems \ref{eq:budget_problem}~and~\ref{eq:cover_problem}. Another factor is related to the connectivity and centrality of nodes from different groups. The solution to the optimization problems  \ref{eq:budget_problem} and \ref{eq:cover_problem} tend to favor nodes which are more central and have high-connectivity. Finally, given the above two factors, we note that the disparity in influence across groups can be further exacerbated for lower values of deadline $\tau$ in the time-critical influence maximization. 

In Figure~\ref{fig:example_graph}, we provide an example to illustrate the disparity across groups in the standard approaches to \tcim. In particular, to show this disparity, we consider the \tcimbudget problem~\ref{eq:budget_problem}, and it is easy to extend this example to show disparity in \tcimcover problem~\ref{eq:cover_problem}. The graph that we consider in this example (see Figure~\ref{fig:example_graph} caption for details) has the two characteristic properties that we discussed above: (i) group $V_2$ is in minority with less than half of the size of group $V_1$, (ii) group $V_1$ has more central nodes compared to group $V_2$, and (iii) nodes in group $V_1$ have higher connectivity than nodes in group $V_2$. We consider the probability of influence in the graph to be $p_e = 0.7$ for all edges, and study the optimization problem~\ref{eq:budget_problem} for budget $B = 2$.
 
For different time critical deadlines $\tau$, we report the following normalized utilities: $\frac{f(S; \Vcal, \Gcal)}{|\Vcal|}$ for the whole population $\Vcal$, $\frac{f(S; \Vcal_1, \Gcal)}{|\Vcal_1|}$ for the group $\Vcal_1$, and $\frac{f(S; \Vcal_2, \Gcal)}{|\Vcal_2|}$ for the group $\Vcal_2$. Here, normalization captures the notion of ``average" utility per node in a group, and automatically allows us to account for the differences in the group sizes. As can be seen in Figure~\ref{fig:example_graph}, the optimal solution to the problem consistently picks set $S = \{a, b\}$  comprising of the most central and high-connectivity nodes. While these nodes  maximize the total utility, they lead to a high disparity in the normalized utilities across groups. As the influence becomes more time-critical, i.e., $\tau$ is reduced, we see an increasing disparity as discussed above. For $\tau=2$, the utility of group $\Vcal_2$ reduces to 0.

\vspace{-3.0mm}
\subsection{Measure of Unfairness}\label{sec:notions:fairnessdef}
Next, in order to guide the design of fair solutions to \tcim problems, we introduce a formal notion of group unfairness in \tcim. In particular, we  measure the (un-)fairness or disparity of an algorithm by the maximum \textit{disparity in normalized utilities} across all  pairs of socially salient groups, given by:
\begin{equation} \label{eq:measure_fair}
\max_{i,j \in \{1, 2, \ldots, k\}} \bigg| \frac{f_\tau(S;\Vcal_{i},\Gcal)}{|\Vcal_{i}|} -  \frac{f_\tau(S;\Vcal_{j},\Gcal)}{|\Vcal_{j}|} \bigg|.
\end{equation}
As discussed above (see Section~\ref{sec:notions:example}), normalization w.r.t. group sizes captures the notion of average utility per node in a group and hence makes the measure agnostic to the group size.  In the next section, we seek to design fair algorithms for \tcim problems that have low disparity (or more fairness) as measured by Eq.~\ref{eq:measure_fair}.

\vspace{-4mm}
\section{Achieving Fairness in TCIM}\label{sec:algorithms}
In this section,  we seek to develop efficient algorithms for \tcim problems under fairness considerations that have low disparity measured by Eq.~\ref{eq:measure_fair} while maintaining high performance.

\vspace{-3mm}
\subsection{Fair TCIM-Budget}\label{sec:algorithms:budget}
\subsubsection{Fairness considerations in \tcimbudget}
A fair \tcim algorithm under budget constraint should seek to achieve the following two objectives: (i) maximizing total influence for the whole population $\Vcal$ as was done in the standard \tcimbudget problem~\ref{eq:budget_problem}, and (ii) enforcing fairness by ensuring that disparity across different groups as per Eq.~\ref{eq:measure_fair} is low.  Clearly, enforcing fairness would lead to a reduction in total influence, and we seek to design algorithms that can achieve a good trade-off between these two objectives.  We formulate the following fair variant of \tcimbudget problem~\ref{eq:budget_problem} that captures this trade-off:
\begin{align*}\label{eq:fair_budget_problem}
&\max_{S \subseteq \Vcal} \underbrace{\sum_i^k {f_{\tau}(S; \Vcal_i, \Gcal)}}_\text{Maximize number of influenced nodes}  \tag{P3} \\
&\text{subject to } \underbrace{|S| \leq B}_\text{Bound seed set size }, \\ &\text{and } \underbrace{\max_{i,j} \Big| \frac{f_\tau(S;\Vcal_{i},\Gcal)}{|\Vcal_{i}|} -  \frac{f_\tau(S;\Vcal_{j},\Gcal)}{|\Vcal_{j}|} \Big| \leq c}_\text{Minimize disparity}
\end{align*}
where $c \in [0,1]$ is a hyperparameter which indicates the maximum level of allowed disparity among the groups. This problem might not be feasible for all the values of $c$. So, one would have to tune this hyperparameter for feasibility and the desired level of disparity. Problem~\ref{eq:fair_budget_problem} has two main objectives, i.e., finding $B$ seeds which will i) \textbf{maximize the total  influence}, which is exactly the same as the traditional influence maximization given in problem~\ref{eq:budget_problem}--- here written as the sum of influences over all the groups, and, additionally, ii) \textbf{minimize the disparity of influence} between different groups up to the prescribed threshold. 

We note that problem~\ref{eq:fair_budget_problem} is NP-Hard and a challenging discrete optimization problem and it does not have the structural properties of submodularity as was the case for the standard \tcimbudget problem~\ref{eq:budget_problem}.

\vspace{-1mm}
\subsubsection{Surrogate \fairtcimbudget with guarantees}

Instead of directly solving problem~\ref{eq:fair_budget_problem}, we introduce a novel surrogate problem that would allow us to indirectly trade-off the two objectives of maximizing total influence and minimizing disparity across groups, as follows:
\begin{align*}\label{eq:surrogate_fair_budget_problem}
  \max_{S \subseteq \Vcal} \sum_{i=1}^k \Hcal(f_\tau(S;\Vcal_{i},\Gcal))  
  \quad \text{ subject to } |S| \leq B, \tag{P4}
\end{align*}
where $\Hcal$ is a non-negative, monotone concave function. 

Optimizing problem~\ref{eq:surrogate_fair_budget_problem} captures both the objectives of the original: i) \textbf{maximizing influence}: since the objective is monotonically increasing it encourages picking more influential nodes, ii) \textbf{minimizing the disparity of influence}: Passing the group influence functions through a monotone concave function $\Hcal$ rewards selecting seeds that would lead to higher influence on under-represented groups early in the selection process; this in turn helps in reducing disparity across groups under the \textit{assumption} that the under-represented groups not only have lower influence in terms of total number of nodes but also have lower influnece in terms of fraction of nodes w.r.t to their groups sizes. In other words, as we are passing the group influences through a \textit{concave} function, the increase in the objective would be higher when under-represented groups are influenced, as demonstrated in figure~\ref{fig:illustration_concave}. 

\begin{figure}
\centering
\includegraphics[angle=0, width=0.7\linewidth]{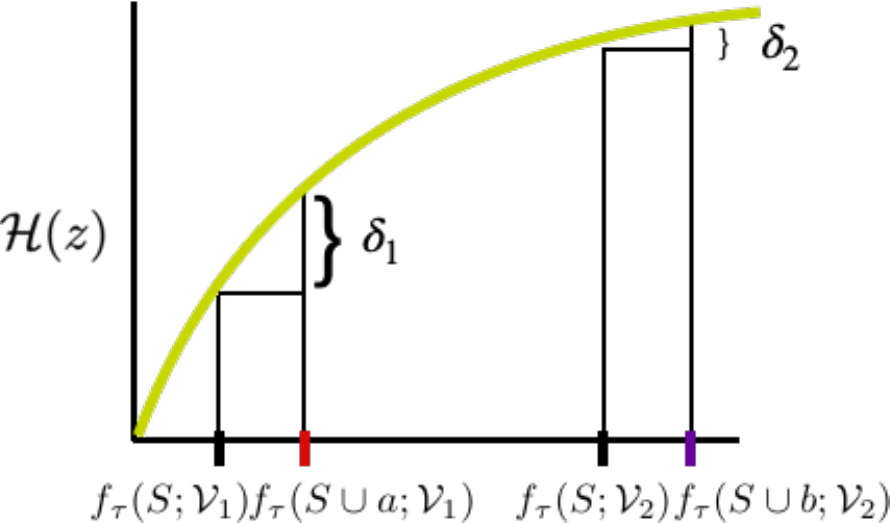}
\vspace{-1.0mm}
\caption{Demonstration of concave function encouraging picking seeds which influence under-represented group. X-axis represents group influence and y-axis represents the value of $\Hcal$ for the corresponding group influence. In this example we have two groups, $\Vcal_1$ and $\Vcal_2$. $\Vcal_1$ is under-influenced compared to $\Vcal_2$, using the seed set $S$. In the next iteration we have an option to either include node $a$ or $b$ in our seed set, both of which add the same amount of total influence. Adding node $a$ in our seed set influences $\Vcal_1$which is the under-influenced group, while adding node $b$ influences nodes from $\Vcal_2$, as demonstrated in the figure. The traditional method, given by problem~\ref{eq:budget_problem}, would treat both of these nodes as equally good. However, since we are passing the group influences through a concave function the increase in the value of $\Hcal(z)$ will be more if we pick node $a$, i.e., our method will pick node $a$ because $\delta_1 > \delta_2$.}
\label{fig:illustration_concave}
\vspace{-4.0mm}
\end{figure}

\xhdr{Trade-off between objectives} It is important to note that controlling the curvature of the concave function $\Hcal$  provides an indirect way to \textit{trade-off} between the two objectives, i.e.,  i) the total influence and ii) the disparity of the solution. For instance, using $\Hcal(z) := \log(z)$ has higher curvature than using $\Hcal(z) := \sqrt z$ and hence leads to lower disparity at the cost of lower total influence (this is demonstrated in the experimental results in Figure~\ref{fig:synth_budget_fairness}). For our illustrative example from Section~\ref{sec:notions}, we report the results for an optimal solution to \fairtcimbudget problem~\ref{eq:surrogate_fair_budget_problem} with $\Hcal(z) := \log(z)$. As can be seen in Figure~\ref{fig:example_graph}, the solution leads to a drastic reduction in disparity across groups for different values of deadline $\tau$ compared to an optimal solution of the standard \tcimbudget problem~\ref{eq:budget_problem} at the cost of reduction in total influence. So, if one wants to penalize disparity of influence more one can pick $\Hcal$ function with higher curvature but at the expense of potentially lower total influence.

While it is intuitively clear that using the concave function $\Hcal(z)$ in problem~\ref{eq:surrogate_fair_budget_problem} reduces disparity, we also need to ensure that the  solution to this problem has high influence for the whole population $\Vcal$ and that the solution can be computed efficiently. As proven in the theorem below,  we can find an approximate solution to problem~\ref{eq:surrogate_fair_budget_problem}, with guarantees on the total influence, by running the greedy heuristic (as was introduced in Section~\ref{sec:background:submodular}).

\begin{theorem} \label{theorem:fair_budget}
Let $\hat{S}$ denote  the output of the greedy algorithm for problem~\ref{eq:surrogate_fair_budget_problem}. Let $S^*$ be an optimal solution to problem~\ref{eq:budget_problem}.  Then, the total influence of the greedy algorithm is guaranteed to have the following lower bound: 
$f_\tau(\hat{S}; \Vcal, \Gcal) \geq (1 - \frac{1}{e})  \cdot \Hcal\big(f_\tau(S^*; \Vcal, \Gcal)\big)$.
\end{theorem}
This is equivalent to the fact that the multiplicative approximation factor of the utility of \fairtcimbudget using greedy algorithm w.r.t. the utility of an optimal solution to \tcimbudget scales as $\big((1 - \frac{1}{e}) \cdot \frac{\Hcal(f_\tau(S^*; \Vcal, \Gcal))}{f_\tau(S^*; \Vcal, \Gcal)}\big)$. Note that as the curvature of the concave function $\Hcal$ increases, the approximation factor gets worse---this further highlights  how the curvature of the function $\Hcal$ provides a way to trade-off the total  influence and disparity of the solution. In the case of $\Hcal(z):= log(z)$, which penalizes the disparity of the solution quite severely due to high curvature, the bound on the total influence achieved by our solution is exponentially related to the optimal solution of  problem~\ref{eq:budget_problem} which does not consider fairness. On the other hand, if $\Hcal(z):= z$, i.e., $\Hcal$ is an identity function, the problem reverts back to problem~\ref{eq:budget_problem}, whose solution might have a higher total influence but could result in high disparity, as evidenced by our experimental results in sections~\ref{sec:eval_synth_budget} and \ref{sec:eval_real_budget}. One can pick $\Hcal$ with the appropriate curvature for the desired level of penalization of the disparity of influence at the cost of total influence. 
\textit{Due to lack of space, the proof of the theorem is included in the appendix.}

\vspace{-3mm}
\subsection{Fair TCIM-Cover}\label{sec:algorithms:cover}

\subsubsection{Fairness considerations in \tcimcover}
A fair \tcim algorithm under coverage constraint should seek to achieve the following two objectives: (i) minimizing the size of the seed set that achieves the desired coverage constraint as was done in the standard \tcimcover problem~\ref{eq:cover_problem}, and (ii) enforcing fairness by ensuring that disparity across different groups as per Eq.~\ref{eq:measure_fair} is low. As was the case for \fairtcimbudget problem above, enforcing fairness would lead to increasing the size of the required seed set, and we seek to design algorithms that can achieve a good trade-off between these two objectives.  We formulate a  fair variant of \tcimcover problem~\ref{eq:cover_problem} that captures this trade-off as follows:
\vspace{-1mm}
\begin{align*}\label{eq:fair_cover_problem}
&\min_{S \subseteq \Vcal} \underbrace{|S|}_\text{Minimize seed set size} \tag{P5}\\
&\text{subject to }  \underbrace{\frac{\sum_i^k f_{\tau}(S; \Vcal_i, \Gcal)}{|\Vcal|} \geq Q}_\text{Bound fraction of influenced node}, \\
&\text{and } \underbrace{\max_{i,j} \Big| \frac{f_\tau(S;\Vcal_{i},\Gcal)}{|\Vcal_{i}|} -  \frac{f_\tau(S;\Vcal_{j},\Gcal)}{|\Vcal_{j}|} \Big| \le c}_\text{Minimize disparity}
\end{align*}
where $c \in [0,1]$ is a hyperparameter, which determines the amount of disparity that is allowed. As in the case of problem~\ref{eq:fair_budget_problem}, it is possible that for some values of $c$ the problem is infeasible. Problem~\ref{eq:fair_cover_problem} has three objectives: i) \textbf{ minimizing size of seed set} that ii) \textbf{influences a prescribed quota} of the population while ii) \textbf{minimizing disparity in the influence} among the groups. 

As in Section~\ref{sec:algorithms:budget}, we note that problem~\ref{eq:fair_cover_problem} is a challenging discrete optimization problem and does not have structural properties as was the case for the standard \tcimcover problem~\ref{eq:cover_problem}.

\vspace{-1mm}
\subsubsection{Surrogate \fairtcimcover with guarantees}
\begin{figure}
\centering
\includegraphics[angle=0, width=0.7\linewidth]{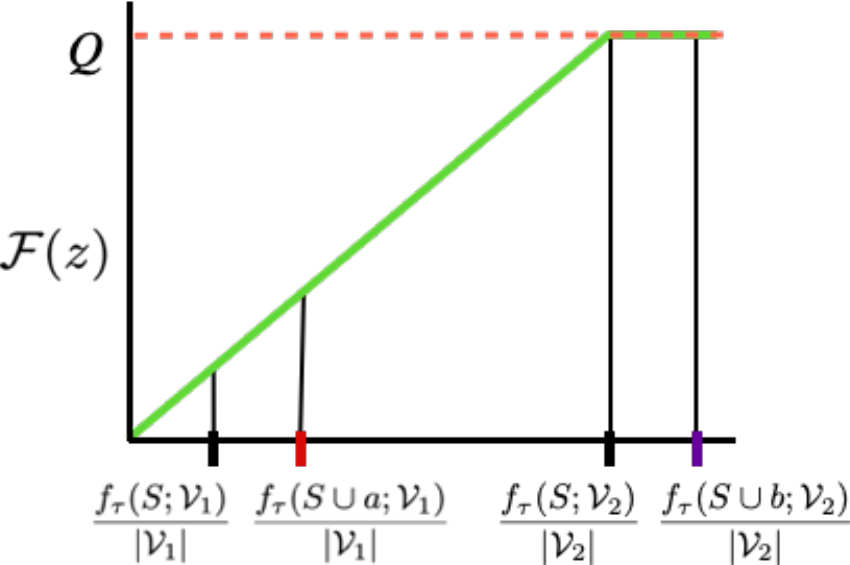}
\vspace{-3.0mm}
\caption{ where $\Fcal(z) = \min\bigg\{\frac{f_{\tau}(S; \Vcal_i, \Gcal)}{|\Vcal_i|}, Q\bigg\} $. Demonstration of the constraint in problem~\ref{eq:surrogate_fair_cover_problem}. X-axis represents the fraction of group influences and y-axis represents the value of per group constraint in problem~\ref{eq:surrogate_fair_cover_problem} for the corresponding group influence. In this example we have two groups, $\Vcal_1$ and $\Vcal_2$ of roughly same size. $\Vcal_1$ has not reached the prescribed quota, $Q$, while $\Vcal_2$ has already been influenced up to the prescribed quota. In the next iteration we have an option to either include node $a$ or node $b$ in our seed set, both of which add the same amount of total influence. Adding node $a$ in our seed set influences only $\Vcal_1$, while adding node $b$ influences nodes from only $\Vcal_2$, as demonstrated in the figure. The traditional method, problem~\ref{eq:cover_problem}, would treat both of these nodes as equally good candidates for including in the seed set because they add equal fraction of total influence. However, since we require all the groups to be influenced up to the required quota, selecting node $a$ will increase our constraint value, $\Fcal(z)$, while by selecting node $b$ the constraint value would stay the same as $\Vcal_2$ has already reached the required quota of influence.}
\label{fig:illustration_cover}
\vspace{-5.0mm}
\end{figure}

Instead of directly solving problem~\ref{eq:fair_cover_problem}, we introduce a novel surrogate problem that indirectly trade-offs the two objectives  of minimizing  the size of selected seed set and minimizing disparity, as follows:
\begin{align*}\label{eq:surrogate_fair_cover_problem}
&\min_{S \subseteq \Vcal} |S| 
\quad \text{ subject to }  \frac{f_{\tau}(S; \Vcal_i, \Gcal)}{|\Vcal_i|}  \geq Q \; \; \; \forall i.\tag{P6}
\end{align*}

Optimizing problem~\ref{eq:surrogate_fair_cover_problem} addresses all the objectives of problem~\ref{eq:fair_cover_problem} by i) \textbf{minimizing the seed set size}, ii) which \textbf{influences all the groups up to the prescribed quota}, $Q$. iii) Thereby, \textbf{disparity} of the feasible solution is \textbf{bounded} by $(1-Q)$. The key idea of using the surrogate objective function in problem~\ref{eq:surrogate_fair_cover_problem} is the following: the problem has a constraint that enforces that at least $Q$ fraction of nodes in each group are influenced by the selected seed set $S$; this in turn directly provides a bound on the disparity of any feasible solution to the problem as $(1 - Q)$. Figure~\ref{fig:illustration_cover} provides a demonstration of the constraints we propose.

While it is intuitively clear that the solution to  problem~\ref{eq:surrogate_fair_cover_problem} reduces disparity, we also would like to bound the size of the final seed set and that the solution can be computed efficiently. As proven in the theorem below,  we can find an approximate solution to problem~\ref{eq:surrogate_fair_cover_problem}, with guarantees on the final seed set size,  by running the greedy heuristic (as was introduced in Section~\ref{sec:background:submodular}).

\begin{theorem} \label{theorem:fair_cover}
Let us denote the output of the greedy algorithm for problem~\ref{eq:surrogate_fair_cover_problem} by set $\hat{S}$.  For group $i \in \{1, \ldots, k\}$, let $S^*_i$ denote an optimal solution to the coverage problem~\ref{eq:cover_problem} for the target nodes set to  $\Vcal_i$, i.e., solving problem~\ref{eq:cover_problem} with constraint given by $\frac{f_{\tau}(S; \Vcal_i, \Gcal)}{|\Vcal_i|} \geq Q$.  Then, the size of the seed set $\hat{S}$ returned by the greedy algorithm is guaranteed to have the following upper bound: $|\hat{S}| \leq \ln(1+ |\Vcal|) \Big(\sum_{i=1}^{k} |S^*_i|\Big)$.
\end{theorem}
\textit{Due to lack of space, the proof of the theorem is included in the appendix.}

\begin{figure*}[t]
 \centering
            \subfloat[Total and group influence]
    {
        \includegraphics[angle=0, width=0.28\linewidth]{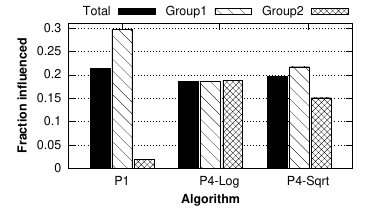}
    \label{fig:synth_budget_fairness}
    }
            \subfloat[Varying budget B]
    {
        \includegraphics[angle=0, width=0.28\linewidth]{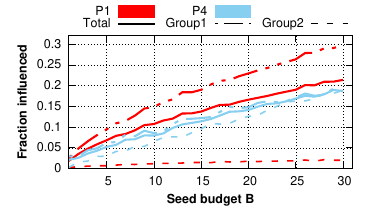}
    \label{fig:synth_budget_cost}
    }
            \subfloat[Varying time deadline $\tau$]
    {
       \includegraphics[angle=0, width=0.28\linewidth]{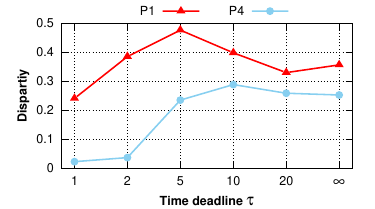}
    \label{fig:synth_budget_deadline}
    }
       \vspace{-1.0mm}    
    \caption{
    [Synthetic Dataset: Budget Problem]     The figures show that solving \tcimbudget problem~\ref{eq:budget_problem} can lead to disparity in number of influenced nodes belonging to different groups, while \fairtcimbudget problem~\ref{eq:surrogate_fair_budget_problem} fares better in terms of achieving parity of influence, with marginally lower total influence. See Section~\ref{sec:eval_synth_budget} for further details.}
    \vspace{-4.0mm}
   \label{fig:synth_budget_1}
\end{figure*}

\vspace{-4mm}
\section{Evaluation on Synthetic Datasets}
\vspace{-1mm}
\label{sec:eval_synth}
In this section, we compare the solutions of different problems on several synthetic datasets. We show that the disparity in influence is affected by varying different properties of the graphs and parameters of the algorithms. 
\vspace{-2mm}
\subsection{Dataset and Experimental Setup}
\label{sec:eval_synth_setup}
First we discuss how we generated the synthetic datasets and then the setup used in our experiments.

\xhdr{Synthetic datasets} We consider stochastic block model to generate the synthetic datasets, particularly we consider an undirected graph with $500$ nodes, where each node belongs to either group $\Vcal_1$ or group $\Vcal_2$. The fraction of nodes belonging to each group is determined by a parameter $g$ (e.g., setting $g= 0.7$ results in $70\%$ of the nodes to be randomly assigned to group $\Vcal_1$). Nodes are connected based on two probabilities: (i) within-group edge probability ({\it Homophily}) $p_{hom}$  and (ii) across-group edge probability ({\it Heterophily}) $p_{het}$. Placing an edge between two nodes goes as follows: given a pair of nodes $(v,w)$, if they belong to the same group, we perform a Bernoulli trial with parameter $p_{hom}$; otherwise we use the parameter $p_{het}$. If the outcome of the trial is $1$, we place an undirected edge $e$ between these two nodes. Each edge has a probability of activation, $p_{e} \in [0,1]$, with which the nodes can activate each other.

\xhdr{Experimental Setup} In our experiments we varied all the aforementioned properties of the graph. We vary each of these graph and algorithmic properties while rest of the properties are set to a default value. We experimented with several default values but as an illustration we include the results for the following default values: $g = 0.7$ yielding $350$ nodes in $\Vcal_1$ and $150$ nodes in $\Vcal_2$. We set  $p_{hom} = 0.025$ and $p_{het} = 0.001$, which yielded $3606$ total edges, out of which $2965$ edges were within group $\Vcal_1$, $514$ within $\Vcal_2$, and $127$ edges connecting nodes across two groups. We used a constant activation probability on all edges given by $p_e = 0.05$. Finally, we consider the time deadline $\tau = 20$, unless explicitly stated otherwise. 
Evaluating utilities, as described in Eq.~\ref{eq:time_critical_utility}, in closed form is intractable, so we used Monte Carlo sampling to estimate these utilities. We used $200$ samples for this estimation, which yielded a stable estimation of the utility function. 
In all the experiments, we pick a seed set by solving the corresponding problem. Then, we use this seed set to estimate the expected number of nodes influenced in the graph using \tcim. We report the following normalized utilities: $\frac{f(S; \Vcal, \Gcal)}{|\Vcal|}$ for the whole population $\Vcal$, $\frac{f(S; \Vcal_1, \Gcal)}{|\Vcal_1|}$ for the group $\Vcal_1$, and $\frac{f(S; \Vcal_2, \Gcal)}{|\Vcal_2|}$ for the group $\Vcal_2$.

\begin{figure*}[t]
 \centering
            \subfloat[Varying influence probability $p_e$]
    {
        \includegraphics[angle=0, width=0.28\linewidth]{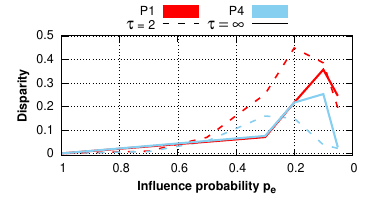}    
    \label{fig:synth_budget_activation}
    }
            \subfloat[Varying group size proportions]
    {
        \includegraphics[angle=0, width=0.28\linewidth]{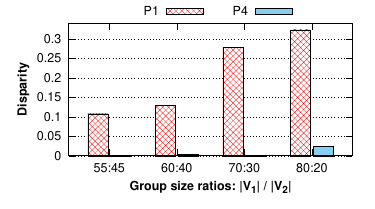}
    \label{fig:synth_budget_groups} 
        }
            \subfloat[Varying cliquishness]
    {
        \includegraphics[angle=0, width=0.28\linewidth]{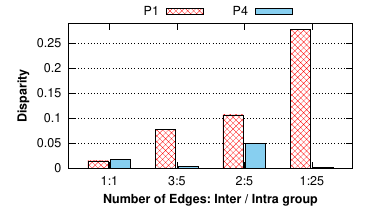}    
    \label{fig:synth_budget_clique}
    }
        \vspace{-1.0mm}
    \caption{
    [Synthetic Dataset: Budget Problem] These figures demonstrate that lower activation probabilities, uneven group sizes, and cliquishness can lead to higher disparity of influence between different groups with \tcimbudget problem~\ref{eq:budget_problem}. In comparison our proposed method, \fairtcimbudget given by problem~\ref{eq:surrogate_fair_budget_problem}, leads to solutions which yield lower disparity. For further details, see Section~\ref{sec:eval_synth_budget}.
    }
    \vspace{-4.0mm}
    \label{fig:synth_budget_2}
    \end{figure*}

\vspace{-3.1mm}
\subsection{TCIM under Budget Constraints}
\label{sec:eval_synth_budget}
Next, we compare the solutions of \tcimbudget problem~\ref{eq:budget_problem} with our solution to \fairtcimbudget problem~\ref{eq:surrogate_fair_budget_problem}, obtained through the greedy algorithm,~\ie, by iteratively picking $B$ seeds which yield maximum marginal gain.
 In all the figures discussed in this section, red color represents the results of \tcimbudget problem \ref{eq:budget_problem}, and blue color represents the results of our solution to the \fairtcimbudget problem \ref{eq:surrogate_fair_budget_problem}. For the experiments in this section, we used a budget of $B = 30$ seeds.
\vspace{-2mm}
\subsubsection{Varying algorithmic properties}\label{sec:algo_prop}
In this section, we vary several properties of the influence maximization algorithm and answer following questions: \\
--- \textbf{Q1}: How does the choice of $\Hcal(z)$ with different curvatures affect disparity and total influence? \\
--- \textbf{Q2}: How does varying seed budget affect disparity? \\
--- \textbf{Q3}: How does varying time deadline affect disparity? \\
--- \textbf{Q4}: How does varying activation probabilities on the edges affect disparity? \\
--- \textbf{Q5}: How effective is our method in reducing disparity? \\
--- \textbf{Q6}: How much cost does our method incur?

\vspace{1mm} \noindent 
\xhdr{[Q1, Q5, Q6] Effect of different $\Hcal(z)$}
Figure~\ref{fig:synth_budget_fairness} presents the comparison of three algorithms: one solving \tcimbudget problem~\ref{eq:budget_problem}, using the greedy heuristic; the other two solving \fairtcimbudget problem~\ref{eq:surrogate_fair_budget_problem}, using two realizations of the concave monotone function, $\Hcal(z)$, given by: (i) $\Hcal(z):= \log(z)$ and (ii) $\Hcal(z):= \sqrt{z}$. Figure~\ref{fig:synth_budget_fairness} shows the fraction of population influenced, both overall and for every group. We can observe that solving the traditional \tcimbudget problem leads to large disparity between the fraction of nodes influenced from each group: while $30\%$ of nodes in group $\Vcal_1$ are influenced, this fraction is only $2\%$ for group $\Vcal_2$. 

On the other hand, our proposed solution to \fairtcimbudget problem results in lower disparity
between the groups, ensuring similar fraction of influenced nodes. We can further see that $\sqrt{z}$, with lower curvature, performs worse than $\log(z)$ in removing the disparity, however incurring lower loss in total influence, as guaranteed by our theoretical results in Theorem~\ref{theorem:fair_budget}. One could consider higher powers of the root to increase the curvature or increase the weights $\lambda$ in problem~\ref{eq:surrogate_fair_budget_problem} for the under-represented group. The \textit{key points} are: i) $\Hcal(z)$ with higher curvature results in lower disparity of influence at the expense of lower total influence. ii) \fairtcimbudget problem results in lower disparity and ii)  the reduction in the total influence is only marginal as guaranteed by Theorem~\ref{theorem:fair_budget}. In the subsequent figures, we only show the results of $\Hcal(z):= \log(z)$ for the solution to problem~\ref{eq:surrogate_fair_budget_problem}. \\\xhdr{[Q2, Q5, Q6] Effect of seed budget}
Figure~\ref{fig:synth_budget_cost} shows the effect of different seed budgets on the number of influenced nodes (from different groups). Dotted and dash-dotted lines correspond to groups $\Vcal_2$ and $\Vcal_1$ respectively, while solid lines represent the total influence. The figure demonstrates that: (i) Disparity in the utility between both the groups increases with the increase in allowed seed budget. A reason for these differences could be the imbalances in groups sizes and average degrees, between both the groups--- $\Vcal_1$ and $\Vcal_2$ comprise $70\%$ and $30\%$ of the nodes respectively. If a very big seed budget is allowed the disparity in influence might also reduce, however in many applications, due to limited resources, it is not practical to have a big budget; (ii) \fairtcimbudget problem results in a lower disparate utility between the two groups compared to \tcimbudget problem; (iii) this reduction in disparity is achieved at a very low cost to the total influence, as guaranteed by Theorem~\ref{theorem:fair_budget}.

\xhdr{[Q3, Q5] Effect of deadline}
Figure~\ref{fig:synth_budget_deadline} compares disparity in the solutions of problems~\ref{eq:budget_problem} and \ref{eq:surrogate_fair_budget_problem} as we vary the value of the deadline $\tau$. Disparity is computed as the absolute difference between the fraction of individuals influenced in each group, given by Eq.~\ref{eq:measure_fair}. The figure demonstrates that: (i) disparity in group utilities does not have a unidirectional trend with increasing time deadline $\tau$. One explanation for the increasing disparity--- for $\tau = \{1,2,5\}$, could be that the seed nodes or the most influential nodes are propagating influence in \textit{both the groups}, but as we increase the time deadline, Group $V_1$, with more nodes and edges, is more efficient at propagating influence compared to Group $V_2$, so it results in a larger disparity. But, after a threshold of increase in $\tau$ both groups are being influenced because longer cascades are allowed. Hence the disparity lowers and then plateaus, for $\tau = \{5,10,20, \infty \}$. One could imagine a case, as shown in the motivating example in Figure~\ref{fig:example_graph}, where seed nodes are surrounded by nodes of \textit{only one group}, in this case increasing time deadline could yield a lower disparity. (ii) Our proposed method, given by problem~\ref{eq:surrogate_fair_budget_problem}, yields solutions which result in much lower disparity.
\begin{figure*}[t]
 \centering
            \subfloat[Greedy selection iterations]
    {
        \includegraphics[angle=0, width=0.28\linewidth]{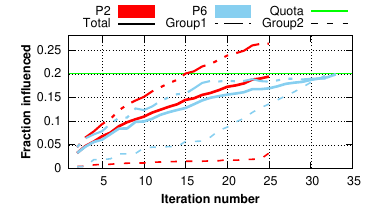}    
    \label{fig:synth_cover_classic_v_fair}
    }
            \subfloat[Varying quota Q: Influence]
    {
        \includegraphics[angle=0, width=0.28\linewidth]{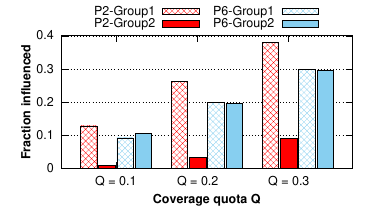}    
        \label{fig:synth_cover_fairness}
    }
            \subfloat[Varying quota Q: Solution set size $|S|$]
    {
        \includegraphics[angle=0, width=0.28\linewidth]{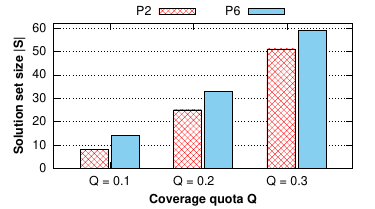}    
        \label{fig:synth_cover_cost}
    }
        \vspace{-1.0mm}
    \caption{
    [Synthetic Dataset: Cover Problem] These figures show a comparison of \tcimcover problem ~\ref{eq:cover_problem}, in red, and \fairtcimcover problem~\ref{eq:surrogate_fair_cover_problem}, in blue. They show that \fairtcimcover achieves lower disparity of influence between different groups     with slightly bigger solution set sizes. See Section~\ref{sec:eval_synth_cover} for further details. 
    }
    \vspace{-4.0mm}
    \label{fig:synth_cover}
    \end{figure*}

\xhdr{[Q4, Q5] Effect of activation probabilities}
Figure~\ref{fig:synth_budget_activation} shows the disparity in influence for different activation probabilities $p_e\in\{0.01,0.05,0.1,$
$0.2, 0.3,0.5,0.7,1.0\}$. The results show that: i) lower activation probabilities could result in larger disparity. This makes intuitive sense, since with lower activation probabilities less nodes have a chance to be influenced. We are using an imbalanced graph, both in terms of group sizes and within and across group connectivity. It is very likely that the seeds selected might belong to the majority group and will have more connections to the nodes from their own group. With low activation probabilities less number of nodes are expected to be influence and the biases in the graph structure would become more pronounced, as evidenced by the results. With the high activation probabilities more number of nodes are expected to be influenced so the disparity in the influence is lower, as demonstrated by the results. ii)  Lower values of $\tau$ tend to have a higher disparity compared to the higher values of $\tau$. The intuition presented in the previous paragraph is confirmed with this experiment. iii) Our method consistently results in a lower disparity. The difference in disparities resulting from the solution of our method compared to the solution of traditional method in more pronounced for lower activation probabilities.

\vspace{-3.3mm}
\subsubsection{Varying graph properties}\label{sec:graph_prop}
In this section, we vary several graph properties and answer following evaluation questions: \\
--- \textbf{Q1}: How does varying group sizes affect the disparity? \\
--- \textbf{Q2}: How does varying connectivity among the groups affect the disparity? \\
--- \textbf{Q3}: How effective is our method in reducing disparity? \\
\xhdr{[Q1, Q3] Effect of group sizes} Figure~\ref{fig:synth_budget_groups} shows the effect of group sizes $g \in \{0.55, 0.6, 0.7,0.8\}$. x-axis represents ratio of the nodes belonging to the two groups and y-axis represents disparity. i) The figure confirms our hypothesis that \textit{imbalance in a graph could lead to disparate influence}, as motivated in the illustrative example given in Figure~\ref{fig:example_graph}. Since we are considering a $1:25$ of $p_{het}:p_{hom}$, i.e., across vs within group edge probability ratios, even slight imbalance in the group sizes could result in a high disparity. The seed nodes or influential nodes are more likely to be from the dominant group and are more likely to be connected with nodes from their own groups. ii) On the other hand our proposed method results in almost no or very little disparity of influence, as it encourages to pick seeds which influence under-represented group. \\
\vspace{1mm}\noindent
\xhdr{[Q2, Q3] Effect of graph connectivity}
Figure ~\ref{fig:synth_budget_clique} demonstrates the importance of the graph structure, particularity connectivity between the two groups, characterized by $(p_{het},p_{hom})\in \{(0.025, 0.025),(0.015, 0.025),$
$(0.01,0.025), (0.001,0.025)\}$. x-axis shows the ratio of across and within group edge probabilities. i) The figure validates our hypothesis that the majority group containing more influential nodes fares better in \tcimbudget problem, as proposed in figure~\ref{fig:example_graph}. Groups $\Vcal_1$ and $\Vcal_2$ comprise $70\%$ and $30\%$ of the nodes, respectively. As we increase the group-preferential attachment, represented by x-axis of figure~\ref{fig:synth_budget_clique}, influential nodes are more likely to have connections within the group $\Vcal_1$, which in turn results in disparate influence propagation. ii) However, our proposed method performs better because it gives less weight to the nodes influenced from the majority group compared to the minority. Hence, our method encourages picking seed nodes which will influence the minority group, as explained in figure~\ref{fig:illustration_concave}. 

\vspace{1mm}
\xhdr{Takeaways} In this section we demonstrated that: (i) solving \tcimbudget problem can lead to disparity of influence in different groups; (ii) the amount of disparity depends on the time deadline, activation probability, relative group sizes, budget, and connectivity of the graph; and (iii) instead, solving \fairtcimbudget results in lower disparity of influence, with marginal reduction in overall influence, as guaranteed by Theorem~\ref{theorem:fair_budget}. 

\vspace{-4.0mm}
\subsection{TCIM under Coverage Constraints}
\label{sec:eval_synth_cover}
Next, we compare solutions of \tcimcover problem \ref{eq:cover_problem}, and our solution to \fairtcimcover problem \ref{eq:surrogate_fair_cover_problem}. We solve both the problems using the greedy algorithm, i.e., iteratively picking seeds which maximize the constraints of problems \ref{eq:cover_problem} and \ref{eq:surrogate_fair_cover_problem} until the required quota is reached. The goal is to reach the prescribed quota $Q$, with minimum number of seeds. In all the figures discussed in this section, red color represents the results of  \tcimcover problem \ref{eq:cover_problem}, and blue color represents the results of our solution to \fairtcimcover problem \ref{eq:surrogate_fair_cover_problem}. We answer the following question in this section: \\
--- \textbf{Q1}: How does our method fare compared to the traditional method over the \textit{iterations of the algorithm}? \\
--- \textbf{Q2}: How effective is our method in reducing disparity for different \textit{reach quotas}? \\
--- \textbf{Q3}: How much cost does our method incur? \\
\vspace{1mm}
\xhdr{[Q1] Effect of iterations}
Figure \ref{fig:synth_cover_classic_v_fair} shows how the fraction of population influenced changes with seed selection at each iteration. Solid lines represent total influence while dash-dotted lines and dotted lines represent groups $\Vcal_1$ and $\Vcal_2$, respectively. In this experiment, $Q$ was set to $0.2$ which is represented by the horizontal green line. The figure demonstrates that: (i) both methods reach the required quota of the population; (ii) however, only the solution set of \fairtcimcover problem~\ref{eq:surrogate_fair_cover_problem} reaches the required quota in both the groups; (iii) while maintaining roughly similar utility for both the groups throughout the iterations; (iv) and it does so \textit{at a small expense of additional seeds}, as guaranteed in Theorem~\ref{theorem:fair_cover}.\\

\begin{figure*}[t]
 \centering
            \subfloat[Total and Group Influence]
    {
      \includegraphics[angle=0, width=0.28\linewidth]{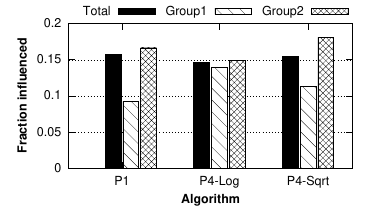} 
    \label{fig:rice_budget_fairness}
    }
            \subfloat[Varying budget B]
    {
      \includegraphics[angle=0, width=0.28\linewidth]{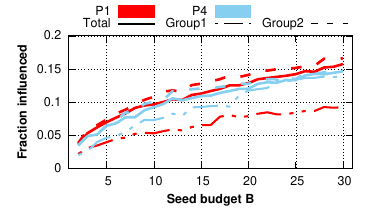} 
    \label{fig:rice_budget_cost}
    }
            \subfloat[Varying time deadline $\tau$]
    {
      \includegraphics[angle=0, width=0.28\linewidth]{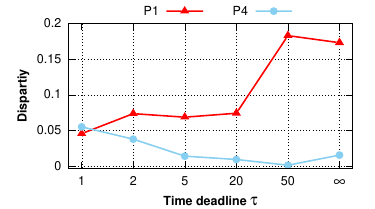} 
    \label{fig:rice_budget_deadline}
    }
        \vspace{-1.0mm}
    \caption{
                [Rice-Facebook Dataset: Budget Problem] Comparison of results solving \tcimbudget problem~\ref{eq:budget_problem} and \fairtcimbudget ~\ref{eq:surrogate_fair_budget_problem}. We experimented with 4 groups and total influence includes all the groups, but we show group influences and disparity for only two groups which showed the maximum disparity. The results demonstrate that our method, given by problem~\ref{eq:surrogate_fair_budget_problem}, yields seed set which propagate influence in a more fair manner, at the cost of a marginally lower total influence. See Section~\ref{sec:eval_real_budget} for further details. 
    }
    \vspace{-4.0mm}
        \label{fig:rice_budget}
    \end{figure*}
\vspace{-4mm}
\xhdr{[Q2, Q3] Effect of quota Q}
Figure~\ref{fig:synth_cover_fairness} shows fractions of individuals that are influenced for different quota $Q$: (i) for different values of the required quota, traditional method given by problem~\ref{eq:cover_problem} results in disparate utility between both the groups which is most likely due imbalance in group sizes and connectivity. (ii) Seeds selected by solving problem~\ref{eq:surrogate_fair_cover_problem} result in a more equal utility because our method explicitly requires every group to be influence up to quota $Q$. Depending on the graph structure, our method could result in a disparity up to $1-Q$. The objective in the constraint given in problem~\ref{eq:surrogate_fair_cover_problem} \textit{only} increases if nodes belonging to the groups are influenced which have not reached the required quota, as demonstrated in figure~\ref{fig:illustration_cover}. A higher disparity between groups could occur when it is not possible to influence the under-influenced group without influencing the already over-influenced group.  In practice a higher disparity could occur, e.g, if one of the groups is very small and very sparsely connected within the group, which is unlikely to occur in practice. (iii) \fairtcimcover problem~\ref{eq:surrogate_fair_cover_problem} uses only a small number of additional seeds, as guaranteed by Theorem~\ref{theorem:fair_cover}.

\xhdr{Takeaways} We compared the result of \tcimcover problem \ref{eq:cover_problem} and our solution to \fairtcimcover problem \ref{eq:surrogate_fair_cover_problem}. The results show that: (i) both methods reach the same fraction of the population; (ii) however, only \fairtcimcover problem results in seed sets influencing the required quota in \textit{all the groups} and results in \textit{a very low disparity} between groups; and (iii) lastly, \fairtcimcover yields \textit{only} slightly larger solution sets as guaranteed by Theorem~\ref{theorem:fair_cover}.

\begin{figure*}[t]
 \centering
            \subfloat[Greedy selection iterations]
    {
        \includegraphics[angle=0, width=0.28\linewidth]{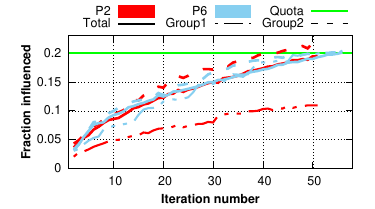} 
    \label{fig:rice_cover_classic_v_fair}
    }
            \subfloat[Varying quota Q: Influence]
    {
        \includegraphics[angle=0, width=0.28\linewidth]{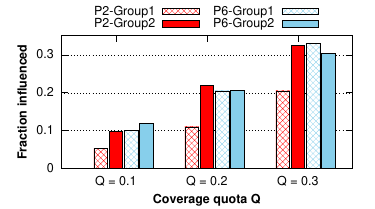} 
    \label{fig:rice_cover_fairness}
    }
            \subfloat[Varying Q: Solution set size $|S|$]
    {
        \includegraphics[angle=0, width=0.28\linewidth]{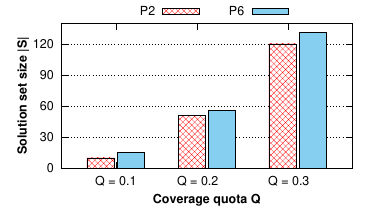} 
    \label{fig:rice_cover_cost}
    }
        \vspace{-1.0mm}
    \caption{
    [Rice-Facebook Dataset: Cover Problem] These figures demonstrate the results of \tcimcover problem~\ref{eq:cover_problem}, in red, and \fairtcimcover problem~\ref{eq:surrogate_fair_cover_problem}, in blue. We experimented with 4 groups and total influence includes all the groups but we show group influences for the two groups which had maximum disparity. The results show that our method achieves a more equal coverage for all the groups at the expense of only slightly larger seed sets. See Section~\ref{sec:eval_real_cover} for further details. 
      }
    \vspace{-4.0mm}
        \label{fig:rice_cover}
    \end{figure*}
\vspace{-2mm}
\section{Experiments on Real-World Datasets}
\label{sec:eval_real}
In this section, we evaluate our proposed solutions using two real-world datasets. We describe the datasets and the details of the experiments, and then present our findings. 

\vspace{-3mm}
\subsection{Dataset and Experimental Setup}
\label{sec:eval_real_setup}
Next, we describe the datasets we used to evaluate our proposed methods, followed by the experimental setup. 

\vspace{1mm}
\xhdr{Rice-Facebook dataset} To evaluate our proposed methods, we used \textit{Rice-Facebook} dataset collected by~\cite{mislove2010you},  where they capture the connections between students at the Rice University. The resulting network consists of $1205$ nodes and $42443$ undirected edges. Each node has $3$ attributes: (i) the residential college id (a number between $[1-9]$), (ii) age (a number between $[18-22]$), and (iii) a major ID (which is in the range $[1-60]$). 

We grouped the nodes (students) into four groups based on their age attributes. We experimented with all four groups while running our algorithms but present the results using only 2 groups which showed the \textit{highest disparity}.
We considered nodes with ages $18$ and $19$ as group $\Vcal_1$ and age $20$ as group $\Vcal_2$. Group $\Vcal_1$ has $97$ nodes and $513$ within-group edges. Whereas, group $\Vcal_2$ has $344$ nodes and $7441$ within-group edges. Overall, there are $3350$ across-group edges going between nodes in $\Vcal_1$ and $\Vcal_2$. 

\xhdr{Instagram-Activities dataset} This dataset was gathered by
\cite{stoica2018algorithmic}. It comprises $553628$ nodes and $652830$ undirected edges. The nodes represent a subset of Instagram users. There exists an edge between two nodes if either of them have liked or commented on each other's photos. Each node has a binary-valued gender attribute, i.e., male or female. $45.5\%$ of the nodes belong to the male group. There are $179668$ within-group edge among males and $201083$ within-group edges among females, while there are $136039$ across-group edges. 

\vspace{1mm}
\xhdr{Experimental Setup} In all the experiments using \textit{Rice-Facebook} dataset, we show the results for activation probability $p_e = 0.01$. All the other parameter were the same as described in section~\ref{sec:eval_synth_setup}. For experiments using \textit{Instagram-Activities} dataset we show the results with activation probability $p_e = 0.06$, time deadline $\tau = 2$, reach quota $Q = \{0.0015, 0.002\}$ and seed budget $B = 30$.  We also experimented with other values of these parameters and get similar results. For \textit{Instagram-Activities} we restrict the seeds to be picked from $5000$ randomly selected nodes from the graph. However the influence was evaluated and propagated on the \textit{entire} network. We used $500$ sample for \textit{Facebook-Rice} dataset and $10000$ samples for \textit{Instagram-Activities} dataset for Monte Carlo estimation of the influence of a node, which yielded very low-variance influence estimates.

\begin{figure*}[t]
 \centering
            \subfloat[Budget Problem: Influence and Cost]
    {
        \includegraphics[angle=0, width=0.28\linewidth]{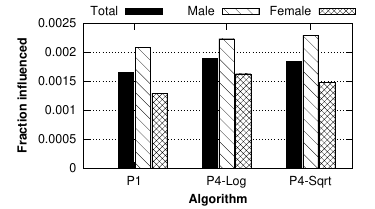} 
    \label{fig:insta_budget_fairness}
    }
            \subfloat[Cover Problem: Influence]
    {
        \includegraphics[angle=0, width=0.28\linewidth]{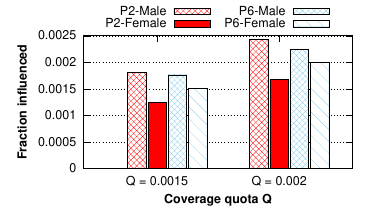} 
    \label{fig:insta_cover_fairness}
    }
            \subfloat[Cover Problem: Cost]
    {
        \includegraphics[angle=0, width=0.28\linewidth]{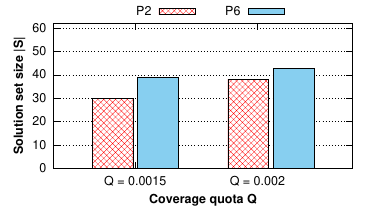} 
    \label{fig:insta_cover_cost}
    }
        \vspace{-2.0mm}
    \caption{
    [Instagram-Activities Dataset] These figures demonstrate a comparison of \tcimbudget vs \fairtcimbudget and \tcimcover vs \fairtcimcover problems. The results show that our methods fare better compared to the traditional methods. Even though the fraction of influence seems small, since the graph comprises 0.5m nodes, the differences in fractions are significant in total numbers.
              }
    \vspace{-4.0mm}
        \label{fig:rice_cover}
    \end{figure*}

\vspace{-4mm}
\subsection{TCIM under Budget Constraint}\label{sec:eval_real_budget}
In this section, we compare the results of \tcimbudget problem~\ref{eq:budget_problem} and our solution to \fairtcimbudget problem~\ref{eq:surrogate_fair_budget_problem}. Red color in all the figures discussed in this section corresponds to the solution of \tcimbudget problem~\ref{eq:budget_problem} and the blue color corresponds to our solution of \fairtcimbudget problem~\ref{eq:surrogate_fair_budget_problem}. In all the experiments in this section we used a seed budget $B = 30$. We answer the following evaluation questions using two \textit{real-world datasets} in this section: \\
--- \textbf{Q1}: How does the choice of $\Hcal(z)$ with different curvatures affect disparity? \\
--- \textbf{Q2}: How does varying seed budget affect disparity? \\
--- \textbf{Q3}: How does varying time deadline affect disparity? \\
--- \textbf{Q4}: How effective is our method in reducing disparity? \\
--- \textbf{Q5}: How much cost does our method incur?  \\
\xhdr{[Q1, Q4, Q5] Effect of different $\Hcal(z)$}
In Figures~\ref{fig:rice_budget_fairness} and \ref{fig:insta_budget_fairness}, we compare the results of \tcimbudget problem~\ref{eq:budget_problem} and \fairtcimbudget problem~\ref{eq:surrogate_fair_budget_problem} using two realizations of $\Hcal(z)$, given by: (i) $\Hcal(z):= \log(z)$ and (ii) $\Hcal(z):= \sqrt{z}$. In Figures~\ref{fig:rice_budget_fairness} the total influence are shown for all the $4$ groups while the group influences are shown for 2 out of the 4 groups which showed the maximum disparity. The results demonstrates that: (i) At a marginal reduction of total influence, as guaranteed by Theorem~\ref{theorem:fair_budget}, our proposed method significantly reduces disparity in influence in case of \text{Rice-Facebook} dataset. However, in the \textit{Instagram-Activities} dataset solving \fairtcimbudget problem results in a \textit{higher} total influence while achieving same or lower disparity for both the groups. This is in line with the finding by~\cite{stoica2019fairness}, which, using this dataset, shows that picking more diverse seeds could increase the total influence compared to greedy degree based seeding strategy. Greedy heuristic is just an approximation of the optimal solution. The optimal solution of the unfair problem cannot yield a lower influence compared to the optimal solution of the fair problem, as it adds additional constraints; (ii) as hypothesized in Section~\ref{sec:algorithms:budget}, a higher curvature function, $\Hcal(z):= \log(z)$, leads to a bigger reduction in disparity compared to $\Hcal(z):= \sqrt{z}$. In \textit{Instagram-Activities} dataset $\Hcal(z):= \sqrt{z}$ does not reduce disparity, however it does result in a higher fraction of influence in under-influenced group. \\
\xhdr{[Q2, Q4, Q5] Effect of seed budget} Figure~\ref{fig:rice_budget_cost} demonstrates the effect of allowed seed budget on the group and total influences. Groups $\Vcal_1$ and $\Vcal_2$ are represented by dash-dotted lines and dotted lines respectively and solid lines correspond to total influence. Similar to the results on synthetic dataset presented in section~\ref{sec:eval_synth_budget}, i) the disparity between the groups seems to increase with increasing budget and ii) our method consistently results in lower disparity for different seed budgets, iii) while incurring a very small cost of total influence. \\
\xhdr{[Q3, Q4] Effect of time deadline} Figure~\ref{fig:rice_budget_deadline} shows the effect of different time deadlines on the disparity between group influences, as calculated by Eq.~\ref{eq:measure_fair}. It demonstrates that: (i) the disparity of influence among groups 
increases as the value of $\tau$ increase, refer to section~\ref{sec:eval_synth_budget} for an intuitive explanation and, ii) our method is very effective in reducing disparity for different values of $\tau$. \\
\xhdr{Takeaways} We demonstrated that: (i) \fairtcimbudget, our proposed method, yields more fair solutions; (ii) this fairness is achieved at a very small reduction of the total influence compared to \tcimbudget problem, as guaranteed by Theorem~\ref{theorem:fair_budget}.

\vspace{-2.5mm}
\subsection{TCIM under Coverage Constraint}\label{sec:eval_real_cover}
Next, we compare \tcimcover problem~\ref{eq:cover_problem} and our solution to \fairtcimbudget problem~\ref{eq:surrogate_fair_cover_problem}. Red color in all the figures discussed in this section corresponds to the solution of \tcimcover problem~\ref{eq:cover_problem} and the blue color corresponds to our solution of \fairtcimcover problem~\ref{eq:surrogate_fair_cover_problem}. We answer the following evaluation question using a \textit{real-world} dataset. \\
--- \textbf{Q1}: How does our method fare compared to the traditional method over the \textit{iterations of the algorithm}? \\
--- \textbf{Q2}: How effective is our method in reducing disparity for different \textit{reach quotas}? \\
--- \textbf{Q3}: How much cost does our method incur?\\
\xhdr{[Q1] Effect of iterations} In Figures~\ref{fig:rice_cover_classic_v_fair} we compare iterations of problem~\ref{eq:cover_problem} and problem~\ref{eq:surrogate_fair_cover_problem}, realized with the log function. 
In each iteration, one seed is selected. Green line represents the required quota of coverage. Dashed-dotted lines, dotted lines and solid lines represent group $\Vcal_1$, group $\Vcal_2$ and total population, respectively. Similar to the results on Synthetic dataset, i) our method consistently results in lower disparity between the two groups, which showed the highest disparity, throughout the iteration of the seed selection algorithm; ii) our method influences all the groups up to prescribed quota; iii) by using small number of additional seeds.\\
\xhdr{[Q2, Q3] Effect of quota} Figures~\ref{fig:rice_cover_fairness}, \ref{fig:rice_cover_cost}, \ref{fig:insta_budget_fairness} and \ref{fig:insta_cover_cost} demonstrate similar results to the synthetic dataset described in section~\ref{sec:eval_synth_cover}. The \textit{keypoint} is that all the groups are covered up to the required quotas with the solution set of \fairtcimcover problem by using only a small number of additional seeds.

\xhdr{Takeaways} We compared the \tcimcover and \fairtcimcover problems in this section using a real world dataset. The results demonstrate that our method is i) effective in reducing disparity ii) by using a small additional number of seeds.

\vspace{-3.5mm}
\section{Conclusions}
In this paper, we considered the important problem of time-critical influence maximization (\tcim) under (i) budget constraint (\tcimbudget) and (ii) coverage constraint (\tcimcover). We showed that the existing algorithmic techniques aimed at maximizing total influence in the population could lead to a huge disparity in utility across the underlying groups. This can put minority groups at a big disadvantage with far-reaching consequences. 
To ensure that different groups are fairly treated, we proposed a notion of fairness and formulated two novel problems to solve \tcim under fairness considerations, namely, \fairtcimbudget and \fairtcimcover. By introducing surrogate objective functions with submodular structural properties, we provided computationally efficient algorithms with desirable guarantees. Experiments over synthetic and real-world datasets demonstrated that our algorithms lead to low disparity in the time-critical influence propagation. This work opens up a variety of new research problems, including extensions to different notions of fairness, considering more complex models of time-criticality in information propagation (such as discounting with time), and developing new optimization methods for solving the fair \tcim problem formulations.

\vspace{-3.0mm}
\section*{Acknowledgements}
 This research was supported in part by a European Research Council (ERC) Advanced Grant for the project “Foundations for Fair Social Computing”, funded under the European Union’s Horizon 2020 Framework Programme (grant agreement no. 789373)

\balance

\appendices

\begin{figure*}[t]
 \centering
            \subfloat[Budget Problem: Influence and Cost]
    {
      \includegraphics[angle=0, width=0.28\linewidth]{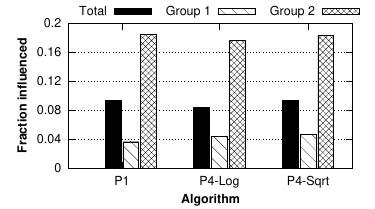} 
    \label{fig:facebook_budget_fairness}
    }
            \subfloat[Cover Problem: Influence]
    {
      \includegraphics[angle=0, width=0.28\linewidth]{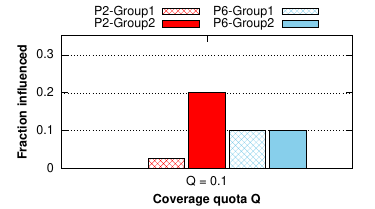} 
    \label{fig:facebook_cover_fairness}
    }
            \subfloat[Cover Problem: Cost]
    {
      \includegraphics[angle=0, width=0.28\linewidth]{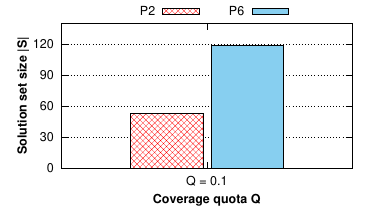} 
    \label{fig:facebook_cover_cost}
    }
        \vspace{-1.0mm}
    \caption{
                [Facebook-Snap dataset] These figures demonstrate a comparison of \tcimbudget vs \fairtcimbudget and \tcimcover vs \fairtcimcover problems. The results show that our method improves the disparity of the influence between different groups. The results for the budget problem show some improvement in the disparity. However, in comparison the reduction in the total influence is also small. One can consider a concave wrapper with a larger curvature to improve the disparity. The results for the cover problem, show a clear improvement in the disparity between the groups. 
    }
    \vspace{-4.0mm}
        \label{fig:facebook_results}
    \end{figure*}
\section{Proof of theorem 1}
\begin{proof}
Since the composition of a non-decreasing concave and a non-decreasing submodular function is submodular \cite{lin2011class}, the objective function in the following problem
\begin{align*}\label{eq:surrogate_fair_budget_problem}
  \max_{S \subseteq \Vcal} \sum_{i=1}^k \Hcal(f_\tau(S;\Vcal_{i},\Gcal))  
  \quad \text{ subject to } |S| \leq B, \tag{P4}
\end{align*}

is monotone submodular function. Let $\tilde{S}$ be the optimal solution and $\hat{S}$ be the output of the greedy algorithm for the problem~\ref{eq:surrogate_fair_budget_problem}, with a fixed budget $B$. Let $S^*$ be an optimal solutions for the following problem
\begin{align}\label{eq:budget_problem}
  \max_{S \subseteq \Vcal}{f_{\tau}(S; \Vcal, \Gcal)} 
  \quad \text{ subject to } |S| \leq B, \tag{P1} 
\end{align}

with a fixed budget $B$. Then, following the standard guarantees of submodular optimization ~\cite{1978-_nemhauser_submodular-max,krause2014submodular} (also see Section 3.4), we get the following bounds

\begin{align}\label{eq:itermediate_1}
\Hcal(f_\tau(\hat{S};V,\Gcal)) &\geq \big(1 - \frac{1}{e}\big) \cdot \Hcal(f_\tau(\tilde{S};V,\Gcal))
\end{align}
Since $\tilde{S}$ is the optimal solution of problem~\ref{eq:surrogate_fair_budget_problem}, given $S: |S| \le B$ following holds, 
\begin{align*}
\Hcal(f_\tau(\tilde{S};V,\Gcal)) &\geq \Hcal(f_\tau(S;V,\Gcal))
\end{align*}
which implies that \begin{align}\label{eq:itermediate_2}
\Hcal(f_\tau(\tilde{S};V,\Gcal)) &\geq \Hcal(f_\tau(S^*;V,\Gcal)).
\end{align}
Combining equations~\ref{eq:itermediate_1} and \ref{eq:itermediate_2} we get 
\begin{align}\label{eq:itermediate_3}
\Hcal(f_\tau(\hat{S};V,\Gcal)) &\geq \big(1 - \frac{1}{e}\big) \cdot \Hcal(f_\tau(S^*;V,\Gcal)).
\end{align}
Since $\Hcal$ is a concave function, 
\begin{align}\label{eq:itermediate_4}
f_\tau(\hat{S};V,\Gcal) &\geq \Hcal(f_\tau(\hat{S};V,\Gcal)).
\end{align}
Combining equations~\ref{eq:itermediate_3} and \ref{eq:itermediate_4} we get 
\begin{align*}
f_\tau(\hat{S};V,\Gcal) &\geq  \big(1 - \frac{1}{e}\big) \cdot \Hcal(f_\tau(S^*;V,\Gcal)),
\end{align*}
which concludes the proof. 
\end{proof}
\\
\\
\\
\\
\\
\\
\\
\\

\section{Proof of theorem 2}

\begin{proof} 
The constraint in the following problem
\begin{align*}\label{eq:surrogate_fair_cover_problem}
&\min_{S \subseteq \Vcal} |S| 
\quad \text{ subject to }  \frac{f_{\tau}(S; \Vcal_i, \Gcal)}{|\Vcal_i|}  \geq Q \; \; \; \forall i,\tag{P6}
\end{align*}

could be rewritten as follows,
\begin{align*}
\sum_{i=1}^k \min\bigg\{\frac{f_{\tau}(S; \Vcal_i, \Gcal)}{|\Vcal_i|}, Q\bigg\} \geq k \cdot Q.	
\end{align*}
The objective function in the constraint is monotone submodular function because monotone submodular functions remain monotone submodular under truncation: if $g(S)$ is monotone submodular so is $f(S):= \min(g(S),c)$ for any constant $c>=0$, and monotone submodular functions are closed under addition \cite{krause2014submodular}. Let $\tilde{S}$ be the optimal solution and $\hat{S}$ be the greedy solution of problem~\ref{eq:surrogate_fair_cover_problem} for a fixed quota $Q$. Let $S^*_i$ be the optimal solution of the following problem:

\begin{align}\label{eq:cover_problem}
 \min_{S \subseteq \Vcal} |S| 
  \quad \text{ subject to }  \frac{f_{\tau}(S; \Vcal, \Gcal)}{|\Vcal|} \geq Q, \tag{P2} 
\end{align}

with target nodes set to $\Vcal_i$ and quota set to $Q$. Then, following the standard guarantees of the submodular optimization \cite{krause2014submodular} (also see Section 3.4) we have the following bound: 
\begin{align}\label{eq:intermediate_5}
|\hat{S}| \leq \ln(1+ |\Vcal|) |\tilde{S}|. 
\end{align}
Since $\tilde{S}$ is the optimal solution of problem~\ref{eq:surrogate_fair_cover_problem}, where all the groups reach the prescribed quota $Q$, $\tilde{S}$ must be at-least as small as any other other set which also reaches all the groups up to the quota $Q$. Hence, 
\begin{align}\label{eq:intermediate_6}
|\tilde{S}| \leq \sum_{i=1}^{k} |S^*_i|.
\end{align}
Combining equations~\ref{eq:intermediate_5} and \ref{eq:intermediate_6} we get  
\begin{align*}
|\hat{S}| \leq \ln(1+ |\Vcal|) \Big(\sum_{i=1}^{k} |S^*_i|\Big),
\end{align*}
which concludes the proof.
\end{proof}

\section{Results Facebook-Snap Dataset}

In this Section, we show the results using the Facebook-Snap dataset proposed by ~\cite{mcauley2012learning}. The dataset comprise $4039$ nodes and $88234$ undirected edges. We used spectral clustering to identify $5$ topological groups in the graph. The five groups comprise $546$, $1404$, $208$, $788$ and $1093$ nodes. We run our algorithms for the entire dataset but report the results only for groups $1$ and $4$, as these groups showed the most disparity in influence using the traditional methods of influence maximization. We used edge weight on $0.01$ and $\tau = 20$. Rest of the parameters were similar to the experiments described in the paper. 

The results are shown in figure~\ref{fig:facebook_results} show that our methods are effective in reducing disparity when considering topological grouping of graphs.

\bibliographystyle{IEEEtran}
\bibliography{fair_inf_max}

\end{document}